\pdfoutput=1
\documentclass[11pt,letterpaper,margin=1in]{article}
\usepackage{booktabs}
\usepackage{fullpage}
\usepackage[utf8]{inputenc}
\usepackage{amsmath, amsthm, thm-restate}
\usepackage{algorithmicx}
\usepackage{xcolor}
\usepackage[ruled,vlined,linesnumbered]{algorithm2e}

\SetAlFnt{\small}
\SetAlCapFnt{\small}
\SetAlCapNameFnt{\small}
\SetAlCapHSkip{0pt}
\IncMargin{-\parindent}
\usepackage{setspace}
\usepackage{mathtools}
\usepackage[numbers]{natbib}
\usepackage{comment} 
\usepackage{tcolorbox} 
\usepackage{xfrac}
\usepackage{hyperref}
\usepackage{multirow}
\usepackage{caption}
\usepackage{bm}
\usepackage{newfloat}
\usepackage{enumitem}
\usepackage{dblfloatfix} 
\usepackage{wrapfig}
\usepackage{fancyhdr}
\newtheorem*{theorem*}{Theorem}
\allowdisplaybreaks
\definecolor{mygreen}{RGB}{10,150,110}
\definecolor{myred}{RGB}{150,10,20}
\hypersetup{
     colorlinks=true,
     citecolor= mygreen,
     linkcolor= myred
}
\renewcommand{\epsilon}{\varepsilon}

\DeclareMathOperator{\E}{\ensuremath{\normalfont \textbf{E}}}

\newcommand{\ascomment}[1]{{\textcolor{purple}{[\textbf{Amin:} #1]}}}
\newcommand{\jccomment}[1]{{\textcolor{blue}{[\textbf{Jiale:} #1]}}}

\newcommand{\hiddencomment}[1]{}

\usepackage[noabbrev,nameinlink]{cleveref}
\crefname{lemma}{Lemma}{Lemmas}
\crefname{theorem}{Theorem}{Theorems}
\crefname{property}{Property}{Properties}
\crefname{claim}{Claim}{Claims}
\crefname{result}{Result}{Results}
\crefname{definition}{Definition}{Definitions}
\crefname{observation}{Observation}{Observations}
\crefname{proposition}{Proposition}{Propositions}
\crefname{assumption}{Assumption}{Assumptions}
\crefname{invariant}{Invariant}{Invariants}
\crefname{line}{Line}{Lines}
\crefname{figure}{Figure}{Figures}
\creflabelformat{property}{(#1)#2#3}
\crefname{equation}{}{}
\crefname{section}{Section}{Sections}
\crefname{appendix}{Appendix}{Appendices}
\crefname{algCounter}{Algorithm}{Algorithms}
\Crefname{algCounter}{Algorithm}{Algorithms}

\newtheorem{theorem}{Theorem}

\newtheorem{lemma}{Lemma}[section]

\newtheorem{claim}[lemma]{Claim}

\newtheorem{observation}[lemma]{Observation}
\newtheorem{remark}{Remark}
\newtheorem*{remark*}{Remark}

\newtheorem{invariant}[lemma]{Invariant}

\definecolor{mylightgray}{RGB}{230,230,230}

\algnewcommand{\IIf}[2]{\textbf{if} #1 \textbf{then} #2}
\algnewcommand{\EndIIf}{\unskip\ \algorithmicend\ \algorithmicif}

\newenvironment{graytbox}{
\par\addvspace{0.1cm}
\begin{tcolorbox}[width=\textwidth,
                  boxsep=5pt,
                  left=1pt,
                  right=1pt,
                  top=2pt,
                  bottom=2pt,
                  boxrule=0pt,
                  arc=0pt,
                  colback=mylightgray,
                  colframe=black,
                  ]
}{
\end{tcolorbox}
}

\newenvironment{whitetbox}{
\par\addvspace{0.1cm}
\begin{tcolorbox}[width=\textwidth,
                  boxsep=5pt,
                  left=1pt,
                  right=1pt,
                  top=2pt,
                  bottom=2pt,
                  boxrule=1pt,
                  arc=0pt,
                  colframe=black,
                  colback=white
                  ]
}{
\end{tcolorbox}
}

\newcounter{algCounter}

\makeatletter
\renewcommand{\paragraph}{%
  \@startsection{paragraph}{4}%
  {\z@}{10pt}{-1em}%
  {\normalfont\normalsize\bfseries}%
}
\makeatother

\makeatletter
\patchcmd{\@algocf@start}
  {-1.5em}
  {0pt}
  {}{}
\makeatother

\title{Stable Matching with Interviews}
\author{
 Itai Ashlagi\\{\em Stanford University} \and 
Jiale Chen\\{\em Stanford University} \and
 Mohammad Roghani\\{\em Stanford University} \and
 Amin Saberi\\{\em Stanford University}
 }
\date{}

\begin{document}
\maketitle
\begin{abstract}
    In several two-sided markets, including labor and dating, agents typically have limited information about their preferences prior to mutual interactions. This issue can result in matching frictions, as arising in the  labor market for medical residencies, where high application rates are followed by a large number of interviews.  Yet, the extensive  literature on two-sided matching primarily focuses on models where agents know their preferences, leaving the interactions necessary for preference discovery largely overlooked. 
    This paper studies this problem using an algorithmic approach, extending Gale-Shapley's deferred acceptance to this context.

    Two algorithms are proposed. The first is an adaptive algorithm that expands upon Gale-Shapley's deferred acceptance by incorporating interviews between applicants and positions. Similar to deferred acceptance, one side sequentially proposes to the other. However, the order of proposals is carefully chosen to ensure an interim stable matching is found. Furthermore, with high probability, the number of interviews conducted by each applicant or position is limited to $O(\log^2 n)$.

    In many seasonal markets, interactions occur more simultaneously, consisting of an initial interview phase followed by a clearing stage. We present a non-adaptive algorithm for generating a single stage set of  in tiered random markets. The algorithm finds an interim stable matching in such markets while assigning no more than $O(\log^3 n)$ interviews to each applicant or position. 
\end{abstract}
\clearpage
\vspace{1cm}
\setcounter{tocdepth}{2} 
\section{Introduction}

In many two-sided matching markets used for dating or allocating labor,  agents often interact prior to forming matches in order to gain information about their potential partners. These interactions, meant for learning preferences, can generate congestion, resulting in market inefficiencies.
For example, in the market for fellowships in sports medicine, on average more than 23 applicants interview for each slot \cite{doi:10.1177/15563316221103585}. Similarly, in dating markets, such as the popular platform Tinder, 
only a small fraction of online conversations culminate in an in-person meeting \cite{lefebvre2018swiping,grontvedt2020hook}.

This paper considers the problem of how to alleviate interview congestion and improve welfare in two-sided matching markets. We take an algorithmic approach to form interviews and find stable matches while minimizing congestion.  

We study this question in the classic two-sided marriage problem by \cite{gale1962college}. 
In the model, there are $n$ applicants and $n$ positions. Each applicant $a_i$ and position $p_j$ has a publicly known value of $u_i$ and $v_j$, respectively. The subjective interest of $a_i$ for matching with $p_j$, 
 denoted by $\epsilon_{ij}^A$ 
 is unknown unless the parties meet each other. The observed utility of $a_i$ in $p_j$ is $v_j + \epsilon^A_{ij}$ if $a_i$ and $p_j$ have met, and equal to $v_j$ otherwise. The distributions of $\epsilon_{ij}^A$s are known, they are all independent, and their expected value is zero. The subjective interest $\epsilon_{ji}^P$ and the observed utility of positions in the applicants are defined similarly.

A matching is interim-stable if (i) all the matched pairs have interviewed each other and (ii) there are no blocking pairs with respect to the observed utilities.  One simple approach for finding an interim stable matching is to ask all pairs to interview each other and then find a stable matching with respect to the observed utilities. That involves $n$ interviews per agent. On the other hand, even in the special case in which all the $u_i$s and $v_j$s are equal to each other, finding an interim-stable matching requires at least $\Omega(\log n)$ interviews per agent. 

The main result of the paper is to show that the number of interviews needed to find an interim-stable matching is much closer to the above lower-bound than the upper bound. In fact, it is possible to find an interim stable matching by assigning only a polylogarithmic number of interviews to each applicant or position. This is shown for both adaptive and non-adaptive models.


\subsection{Approach} 

\paragraph*{Adaptive algorithm:} We present an adaptive algorithm that finds an interim stable matching while requiring that each applicant and position engage in at most $O(\log^2n)$ interviews. The algorithm is applicable to a  general setting, where the publicly known values assigned to applicants and positions can vary arbitrarily, and the subjective interests are independent and identically distributed.

\begin{graytbox}
\begin{theorem*}[See \Cref{thm:query-complexity-two-epsilon}]\label{res:adaptive-algorithm}
    There exists an \textbf{adaptive} algorithm that finds an interim stable matching between $n$ positions and $n$ applicants such that with high probability every applicant and position participates in at most $O(\log^2 n)$ interviews.
\end{theorem*}
\end{graytbox}

The proposed algorithm expands upon deferred acceptance \cite{gale1962college} by dividing the process into distinct interview and proposal stages. Specifically, in the applicant-optimal version, each applicant submits a proposal to the position offering the highest observed utility. The position will accept the ``highest'' proposal or assign an interview if the proposing applicant has the potential to create a blocking pair, taking into account potential changes in the observed value following the interview.

The main technical ingredient of the analysis is to show the probability that an applicant gets rejected by more than $O(\log^2 n)$ consecutive positions is very small. More precisely, with probability at least $1-n^{-3}$, each applicant $a_i$  is matched to position $p_j$ for $j \leq i + O(\log^2 n)$.


\paragraph*{Non-adaptive algorithm:} In the adaptive model, the algorithm is allowed to determine the sequence of interviews based on the results of previous ones. However, in practice, many markets favor a non-adaptive approach in which the interview lists are predetermined beforehand to allow the parties to streamline scheduling and the rest of the logistics. 


 We  present a non-adaptive algorithm for interim stable matching with a poly-logarithmic number of interviews per agent. The algorithm works in tiered markets in which applicants and positions are partitioned into multiple tiers. Two applicants $a_i$ and $a_{i'}$ in the same tier have the same publicly known value, i.e., $u_i = u_{i'}$. On the other hand, if an applicant $a_i$ is in a higher tier than $a_{i'}$, then it is always preferred independent of the interview outcomes, i.e., for all position $p_j$, $\Pr[u_i + \epsilon_{ji}^P > u_{i'} + \epsilon_{ji'}^P] = 1$. 

\begin{graytbox}
\begin{theorem*}[See \Cref{thm:non-adaptive-final}]\label{res:non-adaptive-algorithm}
    There exists a \textbf{non-adaptive} algorithm that, with high probability, finds an interim stable matching for \textbf{tiered markets}  with no more than $O(\log^3 n)$ interviews for each applicant or position.
\end{theorem*}
\end{graytbox}

The design of the non-adaptive algorithm incorporates two key technical components. First, we observe that in markets where all applicants and positions are in a single tier, requiring each applicant to interview a set of $O(\log^2 n)$ positions chosen uniformly at random combined with an applicant-optimal Gale-Shapley algorithm results in an interim stable matching.  

On the other hand, when all positions are in the same tier while applicants are in individual tiers, random interviews prove to be ineffective. 
This is primarily due to positions exhibiting a preference for applicants in the corresponding higher tiers. Consequently, these higher-tier applicants are given higher priority in the matching process. The challenge arises for the lowest-tier applicant who, despite undergoing a poly-logarithmic number of random interviews, struggles to find the sole available position among the interviewed positions. 
To address this, a different approach is taken. By fixing the order of positions, interviews are assigned to each applicant with consecutive positions whose indices closely match the tier number of the applicant. This strategy ensures that applicants with similar tier numbers have a substantial overlap of interviews. Consequently, positions that are interviewed by higher-tier applicants will be matched to them and will not be available options for lower-tier applicants.

The non-adaptive algorithm combines the above two technical ingredients with a few other ideas. Specifically, it creates a sequence of bipartite graphs $G_i(A_i, P_i, E_i)$, where each $G_i$ may be a random graph like the first special case we considered, a sequential graph like the second special case, or a small complete bipartite graph when we have two small tiers. We refer the reader to \Cref{sec:nonadaptive} for the details of the algorithm and its analysis.

Further, in \Cref{sec:extension}, we demonstrate that if applicants and positions have individual preferences within the tiers prior to the interviews, the exact same ideas and proofs presented in \Cref{sec:nonadaptive} carry over, with some mild assumptions about the distribution of individual preferences before the interviews and subjective interest arising from the interview.

The extensive literature on two-sided matching has led to a rich theory and successful designs in practice. However, most of the existing work focuses on models in which agents know their own preferences. Little is known about the interaction period in which agents learn about their preferences by interacting with each other. This paper attempts to address this gap by looking at this problem with an algorithmic lens, focusing on extending Gale-Shapley's deferred acceptance to this setting.

\subsection{Related literature}

Congestion and frictions in the early stages of two-sided markets stem from the lack of complete information and competition. Several approaches have been proposed to improve the screening process, including signaling, limiting the number of applications, and even coordinating interviews.


Signaling allows applicants to send programs ``interest'' signals \cite{lee2015propose,coles2013preference, jagadeesan2018varying}, which intend to improve efficiency by helping in screening applications and interviewing applicants that are typically beyond reach. The literature has investigated how a limited number of signals indicating special interest can enhance efficiency. \cite{jagadeesan2018varying} demonstrate the impact of different signal quantities on the matching outcome.  Signaling is used in the Economics job market \cite{coles2010job} and experimented in residency and fellowship markets \cite{pletcher2022otolaryngology,salehi2019preference}.  Instead of signaling mechanisms, this paper proposes interview mechanisms that address interview congestion by eliminating interviews that are unlikely to form a match.

Another approach considered is to limit the number of applications applicants can send. \cite{agarwal2022stable,beyhaghi2021randomness,skancke2021welfare} demonstrate the benefit of limiting the number of applications when agents have complete information about their own preferences. We note that \cite{beyhaghi2021randomness} and \cite{skancke2021welfare} further assume that agents are ex-ante symmetric.\footnote{Limiting applications has also been proposed in for medical markets \cite{secrest2021limiting, carmody2021application,morgan2021case}.}

Several papers analyzed games induced by inviting agents for interviews 
 \cite{drummond2013elicitation,kadam2021interviewing}, demonstrating various frictions. Other CS papers study how to make interview decisions in the worst-case towards reaching a stable matching \cite{drummond2014preference,rastegari2013two}.

A  related approach for how to coordinate interviews is also studied. Specifically, \cite{manjunath2021interview} and \cite{allman2023interviewing} find that the match rate is high when agents have a similar number of interviews.  \cite{leeandschwarz} propose the idea of incorporating overlaps between positions. Our preference model allows for a more heterogeneous tiered market structure, which results in the non-adaptive algorithm assigning agents possibly an unequal number of interviews.  

Our paper is inspired by \cite{melcher2019may, melcher2019reducing} who propose to conduct an interview match that will limit the number of interviews.   
Similar to our paper, \cite{allman2023interviewing}  study non-adaptive interview mechanisms that generate many-to-many matching in large random markets with  interim stability in the limit.
In their model different applicants (positions) have heterogeneous pre-interview utilities of being matched to some position (applicant). Their heterogeneous utilities are a summation of common scores (generated from a continuous distribution)  and idiosyncratic scores (and in particular the market is not tiered). Instead, we consider arbitrary homogeneous pre-interview utilities (i.e., just common scores) and propose an adaptive algorithm generating an interim stable matching for any finite market size. \cite{allman2023interviewing}  further points to the challenge of handling tiered markets.\footnote{They offer a heuristic for two-tiered markets; this heuristic adds to higher tier agents, safety interviews with lower tier agents.} Instead our non-adaptive algorithm generates an interim-stable matching in every homogeneous tiered market. 





This paper  relies on the existing literature on random two-sided matching markets. Our non-adaptive algorithm incorporates the findings of \cite{ashlagi2017unbalanced} and \cite{CaiThomas22} as a fundamental component. These prior studies specifically examine the average ranking of the matched positions for individual applicants in the applicant-proposing Gale-Shapley algorithm when preference lists are randomized. We further extend the tiered market model in \cite{ashlagi2020clearing} to allow for noisy preferences. \cite{ashlagi2020clearing} focuses on minimizing communication to reach an exact stable matching with high probability in a random market when agents know their own preferences.  Instead,  agents in our model know the public values and have perfect information only over public values and we seek to minimize the number of interviews required to learn about post-interview preferences in order to reach interim stability. 

The concurrent work  \cite{Ashlagi2024interview} extends our framework in two notable ways. First, it examines markets with {\em heterogeneous} pre-interview utilities, addressing scenarios where agents have differing pre-interview preferences. Second, it explores simple decentralized signaling mechanisms to determine which pairs should interview each other. 

The authors demonstrate that in a single-tiered random market with sparse signals ($d = \omega(1)$), almost interim stability can be achieved. Specifically, the matching becomes perfectly interim stable after removing a vanishingly small fraction of agents. Furthermore, in the case of dense signals ($d = \Omega(\log^2 n / p)$), perfect interim stability can be attained in imbalanced markets through short-side signaling.  Additionally, they extend their results to multi-tiered markets and identify conditions under which signaling mechanisms are incentive compatible.

Finally, there is a growing literature in Economics that studies the structure stable matchings when agents have partial knowledge of their own preferences Liu et al. \cite{liu2014stable}, Liu \cite{liu2020stability}. This literature does not consider the matching dynamics that arise prior to being matched. An exception is \cite{echenique2022top}, which contributes to our understanding of the stages that occur prior to a match.

\subsection{Our Model}

Let $A = \{a_1, a_2, \ldots, a_n\}$ and $P = \{p_1, p_2, \ldots, p_n\}$ denote the set of applicants and positions respectively.  The utility of applicant $a_i$ for position $p_j$, $v_{ij}$ can be written as  $v_j + \epsilon^A_{ij}$. The quantity $v_j$ reflects the characteristics of the position, like working hours, salary, or prestige, and is public knowledge.  Applicant $a_i$'s subjective interest in position $p_j$ is captured by $\epsilon^A_{ij}$, which will only be revealed to the applicant after the interview. We assume  $\epsilon^A_{ij}$'s are independently sampled from a known symmetric distribution with mean zero. Define the utility of position $j$ in applicant $i$, $u_{ij} = u_i + \epsilon^P_{ji}$ similarly. Note that $\epsilon^A_{ij}$'s and $\epsilon^P_{ji}$'s are also independent from each other. 

We  define interim stable matchings.  Let the {\em observed utility} of $a_i$ in $p_j$, $v^o_{ij}$ be $v_j + \epsilon^A_{ij}$ if $a_i$ and $p_j$ have interviewed, and equal to $v_j$ otherwise. Define $u^o_{ji}$'s or the observed utility of positions in the applicants similarly. A matching $\mu$ is interim stable if all the matched pairs have interviewed with each other and there are no blocking pairs with respect to observed utilities. In other words,  $\{a_i, p_j\} \in \mu$ and $\{a_k, p_l\} \in \mu$ imply either $v^o_{ij} \geq v^o_{il}$ or $u^o_{ji} \geq u^o_{jk}$. Throughout the paper, we assume the applicants and positions prefer being matched to someone to remaining unmatched.

For ease of notation, we use $\mu(a_i)$ and $\mu(p_j)$ to refer to the position or the applicant they are matched to, respectively.  Define $\mu(a_i) = \emptyset$ if $a_i$ is not matched. We also use $\succ_{a}$ and $\succ_{p}$ to denote the preferences of applicants and positions, e.g., $p_j  \succ_{a_i} p_{j'}$ if and only if $v_{ij}^o > v_{ij'}^o$.

We will study {\em adaptive} and {\em non-adaptive} algorithms for finding interim stable matching. Adaptive algorithms use the outcome of previous interviews to propose the next one. In contrast, non-adaptive algorithms determine the interviews that should be conducted between all the applicants and agents in one shot. 


\section{An Adaptive Algorithm}

In this section, we present an adaptive algorithm for finding an interim stable matching. Our algorithm extends Gale-Shapley's deferred acceptance to incorporate interviews between the applicants and positions. Just like in deferred acceptance, one side makes proposals to the other sequentially, but the order of proposals is chosen carefully so that with high probability, the number of  interviews needed to obtain interim stability is no more than $O(n \log^2 n)$.

\begin{theorem}\label{thm:query-complexity-two-epsilon}
    \Cref{alg:two-epsilon} finds an interim stable matching between $n$ positions and $n$ applicants. Moreover, with high probability, the number of interviews done by each applicant or position is at most $O(\log^2 n)$.
\end{theorem}

Let us start with an informal overview of the algorithm. Consider the applicant-proposing deferred acceptance mechanism. In the beginning, $p_1$ has the largest expected utility and is, therefore, the first choice for all applicants. Similarly, since $a_1$ is the applicant with the largest expected utility, it is $p_1$'s first choice. The algorithm asks $p_1$ and $a_1$ to interview each other. The interview may change $a_1$ and $p_1$'s  preferences. If they are still at the top of each other's preference list, we can match them and reject all other applicants for $p_1$; otherwise, we update the preference list of $p_1$ and $a_1$ and continue the process. 

For an unmatched applicant $a$,  let $\beta(a)$ be the most preferred position from which the applicant is not yet rejected. At any time, $\beta(a)$ is the position to which applicant $a$ is going to propose next. At each step, we consider the  position $p_j$ with the smallest $j$ such that there exists an applicant $a_i$ where $\beta(a_i) = p_j$.



Suppose that among all proposals that $p_j$ is receiving, $a_i$ is the one that has the largest observed utility, i.e. $i = \arg \max_i u^o_{ij}$. If $p_j$ and $a_i$ have already interviewed each other, then the process is similar to deferred acceptance: if $\mu(p_j) \succ_{p_j} a_i $, position $p_j$ rejects applicant $a_i$; otherwise, 
 $p_j$ rejects $\mu(p_j)$ and matches to  $a_i$. 

Now consider the case where  $p_j$ has not interviewed $a_i$. If $\mu(p_{j}) \succ_{p_j} a_i$ and $j < i$,  $p_j$ rejects $a_i$ without an interview. Otherwise, $p_j$ and $a_i$ interview each other and update their observed utility and preference list. The algorithm ends when all agents are matched. See a formal description below.

\begin{algorithm}[H]
    \caption{Adaptive Algorithm for Interim Stable Matching} 
    \label{alg:two-epsilon}
        Initialize $\mu(a) = \emptyset$ for $a \in A \cup P$, and $\forall i,j, v^o_{ij}=v_j, u^o_{ji} = u_i$.
        
        \While{$\exists$ an unmatched applicant}{            
            Let $j^*$ be the smallest index where $\beta(a_i)=p_{j^*}$, for some unmatched applicant $a_i$.
            
            Let $a_{i^*}$ be position $p_{j^*}$'s favorite applicant from the set $\{a_i | \beta(a_i) = p_{j^*} , \mu(a_i) = \emptyset\}$.

                \If{$(a_{i^*}, p_{j^*})$ have not interviewed and $(i^* \leq j^* \quad$ or $\quad a_{i^*} \succ_{p_{j^*}} \mu(p_{j^*}))$ \label{ln:two-epsilon-line-12}}{
                    Position $p_{j^*}$ interviews applicant $a_{i^*}$; update $u^o_{i^* j^*},v^o_{j^* i^*}$. \label{ln:two-epsilon-line-15}

                }
\Else{
                \lIf*{$\mu(p_{j^*})\succ_{p_{j^*}} a_{i^*}$}{
                     $p_{j^*}$ rejects applicant $a_{i^*}$.\\
                }
                \Else{
                
                     $p_{j^*}$ rejects applicant $\mu(p_{j^*})$ if it is not $\emptyset$.
                    
                    $\mu(a_{i^*}) \gets p_{j^*}$ and $\mu(p_{j^*}) \gets a_{i^*}$. \label{ln:matched-if-interivew}
                    
                }             
        }
        }
        
    \Return matching $\mu$.
\end{algorithm}

\subsection{Analysis of \cref{alg:two-epsilon}}


We will prove \cref{thm:query-complexity-two-epsilon} in the rest of this section. We will show that in the course of the algorithm, all positions interview candidates that have a fairly similar global ranking.  More specifically,  with high probability, every position $p_j$ only interviews applicants $a_i$ where $j - O(\log^2 n) \leq i \leq j + O(\log n)$. 

We  start with a few simple observations about the algorithm. First, note that similar to the Gale-Shapley algorithm, when a position $p_j$ gets matched to an applicant $a_i$, it remains matched until the end of the algorithm. Further, $p_j$ may only reject $a_i$ to match with a more preferable applicant. 

\begin{claim}\label{clm:positive-query}

Suppose that unmatched position $p_j$ and applicant $a_i$ interview each other at some point during the algorithm and both observe non-negative $\epsilon^A_{ij}$ and $\epsilon^P_{ji}$. Then, the algorithm (tentatively) matches them to each other. Subsequently,  $p_j$ may only interview applicants $a_{i'}$ for which $i' < \max(i, j + 1)$.

\end{claim}
\begin{proof}
If $\epsilon^A_{ij} \geq 0$ and $\epsilon^P_{ji} \geq 0$, the interview does not change $\beta(a_{i})$ and applicant $a_{i}$ remains the most preferable applicant proposing to $p_{j}$. Hence, the algorithm matches applicant $a_{i}$ to position $p_{j}$ in the next iteration. 

Now, observe that, based on \Cref{ln:two-epsilon-line-12} of \Cref{alg:two-epsilon},  $p_j$ is going to interview an applicant $a_{i'}$ if either $i' \leq j$ or if $u_{i'} >  u_{i} + \epsilon^P_{j i} > u_{i}$, which implies $i' < i$. 
\end{proof}

\begin{observation}\label{obs:consecutive-query}
If position $p_j$ interviews applicant $a_i$ and $i < j$, then  $a_i$ is interviewed by all $p_{j'}$ where $i \leq j' < j$.
\end{observation}
\begin{proof}
Before proposing to $p_j$,  $a_i$  proposes to all positions $p_{j'}$ with $j' < j$. Also, based on \Cref{ln:two-epsilon-line-12} of the algorithm, when $i \leq j'$,  $p_{j'}$ does not reject  $a_i$ without an interview. 
\end{proof}

The next lemma establishes that positions do not interview any applicant with a significantly higher index. Moreover, we show that with the possible exception of $8 \log n$ positions with the highest index (and lowest expected utility for the applicants), the positions get matched sequentially in the increasing order of their index. 

\begin{lemma}\label{lem:few-upward-queries}
With probability of at least $1 - n^{-2}$, if   $p_j$ and   $a_i$ interview each other, then $i \leq j + 8 \log n$. Also, at the time an applicant  $a_{i}$ proposes to   $p_{j}$, 
\begin{itemize}
    \item If $j < n - 8\log n$, then $\mu(p_{j'}) \neq \emptyset$ for $j' < j$,
    \item If $j \geq n - 8\log n$, then $\mu(p_{j'}) \neq \emptyset$ for $j' < n - 8\log n$.
\end{itemize}
\end{lemma}

\begin{proof}
The first part of the lemma holds trivially when $j \geq n - 8 \log n$. So suppose $j < n - 8\log n$ and let $S$ be the set of agents $a_i$ with $i < j + 8 \log n$. At most $j-1$ agents from $S$ can be matched to position $p_{j'}$ for $j' < j$ at any time. Further, applicants propose to positions in the increasing order of their index. Hence, in order for $p_j$ to get matched to an applicant $a_{i'}$ with  $i' > j + 8\log n$, it should reject at least $8 \log n$ applicants from set $S$. We will show that such an event is extremely unlikely. 

    Let $\mathcal{I}$ be the set of the indices of the first $8 \log n$ applicants  interviewed by $p_j$. By the above argument, $\mathcal{I} \subseteq S$. By \Cref{clm:positive-query}, if $p_j$ interviews some applicant $a_i$ and observes that both $\epsilon^A_{ij}$ and $\epsilon^P_{ji}$ are non-negative, then $p_j$ and $a_i$ get matched.  Since $\epsilon^A_{ij}$ and $\epsilon^P_{ji}$ are drawn independently at random from a symmetric distribution with mean zero,  $\Pr[\epsilon^A_{ij} \geq 0 \text{ and } \epsilon^P_{ji} \geq 0] \geq 1/4$. Therefore, 
    \begin{align*}
        \Pr[\not\exists i \in \mathcal{I} \quad \epsilon^A_{ij} \geq 0 \text{ and } \epsilon^P_{ji} \geq 0] = \prod_{i \in \mathcal{I}} \Pr[\epsilon^A_{ij} < 0 \text{ or } \epsilon^P_{ji} < 0] \leq \left(\frac{3}{4}\right)^{8\log n} \leq n^{-3}.
    \end{align*}
    A union bound over the above events implies the probability that there exists  an $i$ such that $a_i$ and $p_j$ interview each other and $i > j + 8\log n$ is at most $n^{-2}$.

The second part of the lemma follows by observing that as long as each position $p_j$ with $j < n - 8 \log n$ gets matched to one of its first $8 \log n$ proposals, $p_{j+1}$ does not receive any proposals before $p_j$ is matched.
\end{proof}

For the rest of the section, we condition on the high probability events stated in \Cref{lem:few-upward-queries}. For $k\in [n]$, let $A_k = \{ a_k, a_{k+1}, \ldots, a_n\}$ and $P_k = \{ p_k, p_{k+1}, \ldots, p_n\}$. Also, let $\bar{A_k} = A \setminus A_k$ and $\bar{P_k} = P \setminus P_k$.

\begin{claim}\label{clm:few-upward-matches}
For any $k \in [n]$, 
\begin{align*}
    \Big\lvert \{(a_i, p_j) \in \mu \lvert i \geq k, j < k\} \Big\rvert \leq 8 \log n.
\end{align*}
\end{claim}

\begin{proof}
    Since we condition on \Cref{lem:few-upward-queries}, 
    $
        \{(a_i, p_j) \in \mu \lvert i \geq k,  j < {k - 8 \log n}\}  = \emptyset.$
    Therefore, 
    \begin{equation*}
        \Big\lvert \{(a_i, p_j) \in \mu \lvert i \geq k, j < k\} \Big\rvert=
        \Big\lvert \{(a_i, p_j) \in \mu \lvert i \geq k, k - 8 \log n \leq j < k\} \Big\rvert
        \leq 8 \log n. \quad \qedhere
    \end{equation*}
\end{proof}

\begin{claim}\label{clm:few-upward-matches2}
Suppose that the algorithm is considering the proposal between applicant $a_{i^*}$ and position $p_{j^*}$.  Then, for any $k \leq j^*$, 
\begin{align*}
    \Big\lvert \{ (a_i, p_j) \in \mu \lvert i < k, j \geq k\} \Big\rvert \leq 8 \log n.
\end{align*}
\end{claim}

\begin{proof}
    We need to prove the statement for  $k < n - 8 \log n$. By \Cref{lem:few-upward-queries}, when the algorithm is considering pair $(a_{i^*}, p_{j^*})$, we have $\mu(p_{j'}) \neq \emptyset$ for $j' < k$. By \Cref{clm:few-upward-matches},
    \begin{align*}
        \Big\lvert \{(a_i, p_j) \in \mu \lvert i \geq k, j < k\} \Big\rvert \leq 8 \log n
    \end{align*}
    which implies,
    \begin{align}
        \Big\lvert \{(a_i, p_j) \in \mu \lvert i < k, j < k\} \Big\rvert \geq (k - 1) - 8 \log n. \label{eq:position-application-downward}
    \end{align}
    On the other hand,
    \begin{align*}
        \Big\lvert\{(a_i, p_j) \lvert i < k, j \geq k\}\Big\rvert &\leq (k - 1) - \Big\lvert\{(a_i, p_j) \lvert i < k, j < k\}\Big\rvert\\
        & \leq 8 \log n    & (\text{By } \Cref{eq:position-application-downward}) \qquad\qquad \qedhere
    \end{align*}
\end{proof}

\begin{lemma}\label{lem:few-downward-queries1}
    With a probability of at least $1-1/n$, if  position $p_j$ interviews applicant $a_i$ then $j \leq i + 2000\log^2 n$.
\end{lemma}

\begin{proof}
    Let $\delta = 8\log n$ and $\gamma = 1999\log^2 n$. The statement holds for $i \geq n - \gamma - \delta$ since $n - i \leq \gamma + \delta < 2000\log^2 n$. Suppose for some $i < n - \gamma - \delta$, applicant $a_i$ and position $p_{i + \gamma + 1}$ interview each other at some point during the algorithm.  At that time, by \Cref{lem:few-upward-queries}, all positions with  index smaller than $i + \gamma + 1$ are matched to an applicant. Also, by \Cref{obs:consecutive-query}, all positions in $\{p_i, p_{i+1},\ldots, p_{i + \gamma}\}$ have already interviewed $a_i$, but none of them is matched to $a_i$. We will show that the probability of such an event is at most $n^{-4}$. 

Consider all positions $p_l$ for $i \leq l \leq i +\gamma$ and define 
    \begin{align*}
        S_l = \{a_{i}, a_{i+1}, \ldots, a_{l + \delta}\} \setminus \{\mu(p_{j'}) \lvert i \leq j' < l \}.
    \end{align*}
By \Cref{clm:few-upward-matches2}, there are at most $\delta$ positions in the set of $\{p_{j'} \lvert i \leq j' < l \}$ which are matched to an applicant with an index smaller than $i$. The rest are matched to an applicant in $\{a_{i}, a_{i+1}, \ldots, a_{l + \delta}\}$.  Therefore, $|S_l| \leq 2 \delta$. 



Define $Y_l$  to be indicating whether $p_l$'s utility for matching to $a_i$ is higher than being matched to all elements in $S_l$ and $X_l$ to be  $Y_l = 1$, $\epsilon^P_{li} \geq 0$, and $\epsilon^A_{li} \geq 0$. By \cref{clm:positive-query}, $X_l = 1$ implies that  $p_l$  can not be matched to an applicant with an index more than $i$. Further, by \Cref{clm:few-upward-matches2} at most $\delta$ positions in $\{p_i, p_{i+1},\ldots, p_{i + \gamma}\}$  may be matched to an applicant with index less than $i$, so it is sufficient to bound the probability that $X 
 = \sum_{j=i}^{i+\gamma} X_j   \leq \delta$.

Observe that $\Pr[Y_l = 1] \geq 1/(2\delta)$ for all $l$. Also, note that because of the independence of the subjective component of utilities, $Y_l$'s are independent. Further, $E[X_l] = 1/4 E[Y_l] \geq 1/(8\delta)$. 
Using Chernoff bound,
    \begin{align*}
        \Pr[X \leq \delta] \leq \Pr[|X - \E[X]| \geq \frac{\gamma}{16\delta}] \leq 2\exp\left(-\frac{\gamma^2/(16\delta)^2}{3 \gamma/(8\delta)}\right) \leq 2\exp\left(- \frac{\gamma}{96\delta} \right)\leq n^{-2}.
    \end{align*}
     The first inequality is due to
    $
        \E[X] - \frac{\gamma}{16\delta}  \geq \frac{\gamma+1}{8\delta} - \frac{\gamma}{16\delta} > \frac{\gamma}{16\delta} \geq \delta,
    $ 
    assuming $n$ is sufficiently large. 
    Applying  union bound over all applicants completes the proof.
\end{proof}

\begin{proof}[Proof of \Cref{thm:query-complexity-two-epsilon}]
    The statements of \Cref{lem:few-upward-queries} and \Cref{lem:few-downward-queries1} hold with high probability, for all applicants $a_i$. Therefore, every $a_i$ is interviewed only by positions $p_j$ where $j \in [\max(1, i - 8 \log n), i + \min(n, 2000\log^2 n)]$. Similarly, for every position $p_j$  in \Cref{alg:two-epsilon}, the position only interviews  $a_i$'s for which $i \in [\max(1, j - 2000\log^2 n), \min(n, j + 8\log n)]$. Therefore, the number of interviews done by each applicant or position is at most $O(\log^2 n)$.

    The proof of stability is fairly similar to the analysis of deferred acceptance. Suppose
    there exists a blocking pair $(a_i, p_j)$. This implies that $a_i \succ_{p_j} \mu(p_j)$  and $p_j \succ_{a_i} \mu(a_i)$. Therefore, $a_i$ must have been rejected by $p_j$ during one of the iterations of \Cref{alg:two-epsilon}.
    This implies that position $p_j$ is matched with an applicant who has a higher observed utility, and the current match of $p_j$ should not be worse than $a_i$, i.e., $\mu(p_j) \succ_{p_j} a_i$. That is a contradiction, as we initially assumed that $(a_i, p_j)$ is a blocking pair. 
\end{proof}

\section{A Non-Adaptive Algorithm for Tiered Markets}

The algorithm in the previous section forms the sequence of interviews adaptively, suggesting each interview based on the outcome of the earlier ones. However, many markets operate in a more simultaneous manner, where there is a stage of interviews followed by a clearing phase. Therefore, it may be more efficient if the interview lists are decided either in advance or independently for each position.

In this section, we present a non-adaptive algorithm for the problem in markets in which applicants and positions are partitioned into tiers, extending the model in \cite{ashlagi2020clearing} to allow for post interaction noise. Applicants have the same ex-ante utility for two positions in the same tier but prefer a position in a higher tier to a lower-tier position with probability 1. The same is true for positions.

More formally, let $0 = \tau_0 < \tau_1 < \ldots < \tau_M=m$ denote the tiers of positions.  Remember the utility of an agent $a_i$ for position $p_j$, $v_{ij} = v_j + \epsilon^A_{ij}$.  
We will assume that $v_j = v_k$ if and only if $p_j$ and $p_k$ are in the same tier, i.e.,  $j,k \in [\tau_{l} + 1, \tau_{l+1}]$ for some $l$, 
and $v_j >> v_k$ if $p_j$ is in a higher tier than $p_k$, i.e. $j \leq \tau_l < k$ for some $l$.  Further, $\epsilon^A_{ij}$'s are all independent and identically distributed, and their distributions are bounded and symmetric around $0$. Therefore,  $\Pr[v_{ij} \geq v_{ik}] = 1/2$ if $p_{j}$ and $p_k$ belong to the same tier and equal to 1 otherwise. We use {\em relative index} of a position to show the index of the position within the tier it belongs to. Formally, for position $p_j$ where $ \tau_{l-1} < j  \leq \tau_l$, the relative index is equal to $j - \tau_{l-1}$. Define the tiers $0 = \gamma_0 < \gamma_1 < \ldots < \gamma_N=n$ for applicants similarly. We assume that the number of applicants and positions are equal in the general tiered markets, i.e. $n = m$, and every applicant and position will be matched. However, to develop the algorithm for the general tiered market, we use two subroutines that might need to solve a subproblem with $n\neq m$ case (see \cref{sec:example-one} and \cref{sec:example-two}).



Our non-adaptive algorithm works in two phases. In the first phase, the algorithm proposes a set of interviews between the two sides. We represent this by a bipartite graph $G(A, P, E)$, where each edge $e\in E$ denotes an interview between its endpoints. In the second phase, after the interviews are conducted and the corresponding $\epsilon_{ij}^A , \epsilon_{ji}^P$  for  $(a_i, p_j) \in E$ are revealed, the algorithm finds an interim stable matching between two sides. The main result of this section is that the proposed interviews in the first phase are such that  (i) no agent or position has more than $O(\log^3 n)$ interviews, and (ii) after the interviews are done, the algorithm can find an interim stable matching supported by $G(A, P, E)$ with high probability.

In order to explain the main ideas of the algorithm, it will be instructive to look at two special cases in the following subsections.

\subsection{Example 1: Single Tier Structure}\label{sec:example-one}
In the first example, the applicants and positions are each in their own single tier, i.e. $N = M = 1$. The number of applicants and positions can be different\footnote{Note that in the setting of the general tiered market, we assume $n=m$, however, the single tier structure serves as a subroutine of the algorithm for the general tiered market and might need to solve a subproblem with $n\neq m$.} but assume that $n$ and $m$ are sufficiently large, and without loss of generality, $n \leq m$. In this case, first form $G$ by connecting each applicant to $\delta$ positions chosen uniformly at random from the $P'=\{p_1, p_2, \ldots, p_n\}$.

Each applicant $a_i$ interviews with $\delta$ positions and the distribution of $\epsilon_{ij}^A$'s are symmetrically distributed around 0. So, for every $i$,   $\delta/2$ of the $\epsilon_{ij}^A$'s are going to be non-negative in expectation. In fact, a simple Chernoff bound shows that with high probability, for every applicant $a_i$, at least $\delta/3$ of positions $p_j$ interviewing the applicant have a non-negative $\epsilon_{ij}^A$.

\begin{restatable}{claim}{positiveCounts}\label{lem:positive-counts}
    Among $\delta$ positions $p_j$ that applicant $a_i$ is interviewing, at least $\delta/3$ of them have $\epsilon_{ij}^A \geq 0$ with probability $1-\exp(-\delta/36)$.
\end{restatable}

Let $\pi_i^A$ be the preference list of applicant $a_i$ in the decreasing value of $\epsilon_{ij}^A$ (break ties randomly). It is not hard to see that each $\pi_i^A$ is a random permutation chosen over the ordering of positions. Following the same argument, we can form  $\pi_j^P$ to be the random permutations representing the preference list of position $p_j$ according to the decreasing values of $\epsilon_{ji}^P$. The above observation allows us to take advantage of the following lemma.

\begin{lemma}[\cite{pittel1992likely}, Theorem 6.1]\label{lem: worst-case-ranking}
    Suppose that we run the applicant-proposing Gale-Shapley algorithm for the above random permutations. With probability $1-O(m^{-c_t})$, each applicant $a_i$ gets matched to a position with a rank less than $(2+t)\log^2 m$ in permutation $\pi_i^A$, where $c_t=2t[3+(4t+9)^{1/2}]^{-1}$.
\end{lemma}

The above lemma combined with \Cref{lem:positive-counts} implies that every applicant $a_i$ is matched with a position $p_j$ where $\epsilon_{ij}^A > 0$.


\begin{restatable}{lemma}{finalBoundSingleTier}\label{lem:final-bound-single-tier}
    For any $t$ and $\delta$ such that $\delta\geq 3(2+t)\log^2 m$, with probability $1-\exp(-\delta/36)-O(m^{-c_t})$, each applicant $a_i$ gets matched to a position $p_j$ that has interviewed $a_i$ and $\epsilon_{ij}^A > 0$, where $c_t=2t[3+(4t+9)^{1/2}]^{-1}$.
\end{restatable}

\subsection{Example 2: One Large Tier vs Multiple Singleton Tiers}\label{sec:example-two}

In this example, all positions belong to the same tier, but each applicant is in a tier of its own. In other words, the applicants are ex-ante indifferent among positions, but all positions  prefer applicant $a_i$ to  $a_j$ when $i < j$. 

Before presenting our solution, it may be worthwhile to understand why forming the interview graph randomly in the same fashion as our previous example is not suitable for this variation of the problem. First, note that since agent $a_1$ is the first choice for all the positions,  $a_1$ should be matched to its most preferred position in every interim stable matching. Removing applicant $a_1$ and that position, the same argument applies to applicant $a_2$, and so on.

Now suppose that  each applicant interviews $\delta$ random positions similar to the algorithm in \Cref{sec:example-one}. Since we choose $\delta$ random positions for applicant $a_1$ to interview and then applicant $a_1$ chooses the most preferred position after the interviews, all positions have the same probability of being the most preferred position for applicant $a_1$ in this process. Therefore, the process is equivalent to the case that we choose a random position for applicant $a_1$ and remove both, then we choose a random position for $a_2$ and remove both, and so on. Suppose the algorithm successfully matches the first $n-1$ applicants to their positions. When only the last applicant $a_n$ remains, it is crucial to show that the only remaining position should be one of the positions that $a_n$ interviewed. Otherwise, the resulting matching is not interim stable. Since we are removing positions randomly, the probability that the last remaining position is among the $\delta$ position that $a_n$ interviewed is $\delta/m$. Therefore, we cannot find an interim stable matching supported within $G$ if we want $\delta=poly(\log m)$.

Instead, for this example, our algorithm forms $G$ by connecting every agent $a_i$ to positions $p_j$ such that $j \in [\max(1, i-\theta), \min(n, i + \delta)]$, where $\theta = \Theta(\log^3 m)$ and $\delta = \Theta(\log^2 m)$.  Note that the degree of applicants with an index close to 1 or $n$ may be smaller than the rest of the applicants. 

In this construction, applicants whose tiers are close to each other interview for sets of positions that are similar. Therefore, there is a high correlation between the positions to which they get matched. As a result, with high probability, for every $i$, the set of positions that have interviewed $a_i$  and are not matched to one of the higher-tier applicants is non-empty. Further, $a_i$ can find a stable match in this set. We will prove this statement in more generality in the next section when analyzing the algorithm for the general case.




\subsection{An Algorithm for General Tiered Markets}
\label{sec:nonadaptive}

We are now ready to present our algorithm and its analysis for tiered markets in their full generality. As we said before, the algorithm works in two phases. In phase 1,  \Cref{alg:general-tiers-interview} constructs a bipartite interview graph $G(A, P, E)$ between applicants and positions. This is done possibly in several iterations. In each iteration, $i$, the algorithm adds an edge set $E_i$ between two subsets $A_i \subseteq A$ and $P_i \subseteq P$ to $E$. Further, it identifies whether the applicants or positions are going to be the proposing side in this part of the graph. Note that  an applicant or a position can be in multiple $A_i$s or $P_i$s. 

In the second phase, \Cref{alg:general-tiers-stable} conducts the interviews between all the pairs in $E$ and updates the preferences of the two sides. Then, in each iteration $i$, it removes vertices that are matched in earlier iterations and implements either a position or applicant-proposing deferred acceptance algorithm between $A_i$ and $P_i$. We discuss the details of this algorithm later in the section. 

In practice, one should first implement \Cref{alg:general-tiers-interview} to construct the interview graph and then \Cref{alg:general-tiers-stable} to implement the interviews and find the stable matching. But it will be useful for the analysis to couple the two algorithms and consider them together after each iteration $i$. 




\Cref{alg:general-tiers-interview} considers applicants and positions starting from the top. Let $X$ and $Y$ be the set of top-tier applicants and positions. Label the sets $X$ and $Y$ in such a way that $|X| \leq |Y|$. 
We will consider three different cases: 
 (1) If both $|X|$ and $|Y|$ are relatively small, i.e. $|X| \leq \delta$ and $|Y| \leq \delta$ for $\delta = \Theta(\log^2 n)$, we can afford to set up an interview for each pair between $X$ and $Y$. (2) If  $|X| \leq \delta$ and $|Y| > \delta$, we choose a set $S$ of size $\delta$ comprised of vertices of $Y$ with the lowest index and set up an interview between each pair in $S$ and $X$. (3) If both $|X|$ and $|Y|$ are large than $\delta$, we use the same approach as \Cref{sec:example-one}. Specifically, for each vertex in $X$, we choose $\delta$ random vertices from the first $|X|$ positions in $Y$ to interview.

When both $|X|$ and $|Y|$ are larger than $\delta$, we can remove both $X$ and the first $|X|$ vertices of $Y$ from $A$ and $P$ and move to the next iteration of the algorithm. The situation is more subtle when for a short tier $X$, we choose a set $S \in Y$ which has a size larger than $|X|$. In this case, in the same iteration of \Cref{alg:general-tiers-stable}, some of the vertices of $S$ will remain unmatched. A priori, and without knowing the outcome of the interviews between these two sets, we do not know which vertices of $S$ are going to remain unmatched. Therefore, we will not remove any vertex from $Y$ immediately. Because of this subtlety, we will have to keep track of the ``effective cardinality'' $e(X)$ and $e(Y)$. They are updated to be equal to the number of unmatched vertices in $X$ and $Y$, respectively, at the same iteration in \Cref{alg:general-tiers-stable}.

The algorithm continues in the same fashion. At every step, the top non-empty tiers $X$ and $Y$ from each side are selected and named so that $e(X) \leq e(Y)$. We call $X$ the {\em short side} and $Y$ the {\em long side}.  

\begin{algorithm}[H]
    \caption{The Algorithm for Constructing the Interview Graph}
    \label{alg:general-tiers-interview}
        Let $\delta = 36\log^2 n$ and $\theta = 72\log^3 n$.

        Initialize $e(X)=|X|$ for all tiers $X$ as the number of unmatched vertices in the tier.

        Initialize $G(A, P, E)$ to be an empty graph, $D \gets []$, and $k \gets 1$.

        \While{$A \neq \emptyset$ and $P \neq \emptyset$}{
            Consider the top non-empty tiers of $A$ and $P$ and label them $X$ and $Y$ s.t. $e(X) \leq e(Y)$. \label{ln:top-tiers-selection}

            \If{$e(Y) \leq \delta$}{
                Let $E_k$ be the set of all edges between $X$ and $Y$. 
                \algorithmiccomment{Case 1}
                
                $e(Y) \gets e(Y) - e(X)$.
            }
            \ElseIf{$e(X) \leq \delta$}{
                Let $S$ be the set of $\delta + |Y| - e(Y)$ vertices in $Y$ with the lowest index.
                \algorithmiccomment{Case 2}
                    
                Let $E_k$ be the set of all edges between $X$ and $S$.

                $e(Y) \gets e(Y) - e(X)$.
            }
            \Else{
                Let $S$ be the set of $e(X) + |Y| - e(Y)$ vertices in $Y$ with the lowest index.
                \algorithmiccomment{Case 3}

                Form  $E_k$ by connecting each $x \in  X$ to $\delta + |Y| - e(Y)$ random vertices of $S$.

                Remove all vertices of $S$ from $Y$ and then $e(Y) \gets e(Y)-e(X)=|Y|$. \label{ln:line-16-case3}
            }

            \lIf{$X$ is in applicant side}{$D(k) \gets \text{'applicant proposing'}$}
                \lElse{$D(k) \gets \text{'position proposing'}$}

            Remove all vertices of $X$ and $e( X) \gets 0$. \label{ln:remove-short-side}

            \lIf{$e(Y) = 0$}{remove all vertices of $Y$}
            
            \lIf{$|Y| - e(Y) \geq \theta$}{remove the $|Y| - e(Y) - \theta$ vertices of $Y$ with lowest indices} \label{ln:removing-vertices}
            
            $k \gets k + 1$.
        }

    \Return $G(A, P, E = \bigcup_{i < k} E_i), D, k$.
\end{algorithm}



In the rest of this section, we prove the correctness of the algorithm and give an upper bound on the number of interviews done by every applicant or position. First, we show that at any time during the algorithm $e(X)$ is close to $|X|$.

\begin{invariant}\label{inv:few-not-expired}
    For a tier $X$, after every iteration during the course of \Cref{alg:general-tiers-interview}, $e(X) \leq |X| < e(X) + \theta$.
\end{invariant}
\begin{proof}
We prove this by induction on the number of iterations. Initially, $|X| = e(X)$. Now assume that $e(X) \leq |X| < e(X) + \theta$ before the $i$th iteration. If $X$ is the short side in iteration $i$, then all vertices of $X$ will be removed after this iteration, and thus $|X| = e(X) = 0$. Also, if $X$ is the long side and the algorithm is in case 3, then $|X| = e(X)$ by \Cref{ln:line-16-case3} of \Cref{alg:general-tiers-interview}. 

In all other cases, $e(X)$ decreases after the iteration. Further, by \Cref{ln:removing-vertices}, if $|X| - e(X) \geq \theta$, we  remove vertices of $X$ for which $|X| < e(X) + \theta$.
\end{proof}

Also, it is not hard to see that all applications and positions will be removed by the end of \Cref{alg:general-tiers-interview}.

\begin{observation}\label{obs:all-removed}
After the last iteration of \Cref{alg:general-tiers-interview}, all applicants and positions are removed.
\end{observation}
\begin{proof}
    Let $X_1, \ldots, X_{\tau_M}$ be all tiers of positions and $Y_1, \ldots, Y_{\tau_N}$ be all tiers of applicants. Note that $\sum_{i=1}^{\tau_M} e(X_i)= \sum_{i=1}^{\tau_N} e(Y_i)$ at all times during the execution of \Cref{alg:general-tiers-interview} since in all three cases, the effective cardinality of the short side tier is subtracted from both $\sum_{i=1}^{\tau_M} e(X_i)$ and $\sum_{i=1}^{\tau_N} e(Y_i)$. Since the algorithm does not terminate until either or both $\sum_{i=1}^{\tau_M} e(X_i) = 0$ or $\sum_{i=1}^{\tau_N} e(Y_i) = 0$, then all vertices are removed when the algorithm terminates.
\end{proof}


\begin{algorithm}[H]
    \caption{Algorithm for Finding a Stable Matching after the Interviews}
    \label{alg:general-tiers-stable}
    Let $G(A,P,E), D,$ and $k$ be the output of \Cref{alg:general-tiers-interview}.

    Interview all edges of $G$.

    \For{$i = 1$ to $k - 1$ \label{ln:parallel-forloop}}{
        \If{$D_i = $ 'applicant proposing'}
        {Run applicant proposing Gale-Shapley on the unmatched endpoints of $E_i$.}
        \Else{Run position proposing Gale-Shapley on the unmatched endpoints of $E_i$.}
    }

    \Return the matching.

\end{algorithm}

We let $A_i$ be the set of applicants that are endpoints of $E_i$. We define $P_i$ similarly. Moreover, let $A'_i \subseteq A_i$ and $P'_i \subseteq P_i$ be the set of applicants and positions that are endpoints of $E_i$, and are unmatched when we run Gale-Shapley in \Cref{alg:general-tiers-stable} between $A_i$ and $P_i$. Therefore, if $A_i$ and $P_i$ belongs to tiers $X$ and $Y$, then $|A_i'| = e(X)$ and $|P_i'| = e(Y)$ at the time of iteration $i$. 
As we mentioned before, a vertex may appear in multiple $A_i$s or $P_i$s. We will show that such sets are consecutive. To do so, we first give a helpful property of the algorithm.

\begin{claim}\label{clm:covered-indices}
    For a tier $X$ that has interview with tiers $Y_1,Y_2,\dots,Y_r$ where $e(Y_i)\leq \delta$ for all $i$, define $X^{j-1}$ be tier $X$ before interviewing with $Y_j$ in \Cref{alg:general-tiers-interview} for $j \in [1,r]$. Suppose $e(X^{j-1})>\delta$ for $j\in[1,r]$, then all vertices in $X^{j-1}$ with relative indices at most $S(X)+\delta-e(X^{j-1})$ will be selected for interview with $Y_j$, and $e(X^j)=e(X^{j-1})-e(Y_j)$, where $S(X)$ is the initial size of tier $X$ before the execution of \Cref{alg:general-tiers-interview}.
\end{claim}
\begin{proof}
    The algorithm always chooses vertices with the lowest index from the long side to interview. Consider the interview with $Y_j$, those are the vertices with a relative index of $S(X) - |X^{j-1}| + 1$ to $(S(X) - |X^{j-1}|) + (\delta + |X^{j-1}| - e(X^{j-1})) = S(Z) + \delta - e(X^{j-1})$. And by definition of \Cref{alg:general-tiers-interview}, $e(X^j)=e(X^{j-1})-e(Y_j)$. 
\end{proof}

\begin{claim}\label{clm:consecutive-queries}
    For each vertex $v$, there exist $L$ and $R$ such that $v$ is in $A_i$ or $P_i$ iff $L \leq i \leq R$.
\end{claim}
\begin{proof}
    Let $Z$ be the tier containing $v$ and $S(Z)$ be the initial size of tier $Z$ at the start of \Cref{alg:general-tiers-interview}. Also, let $L$ be the first iteration considering $v$ and $e(Z)$ be the effective size of tier $Z$ at that point. It is not hard to see that if $Z$ is the short side in this iteration, or if the short side $Z^\prime$ has size $e(Z^\prime)>\delta$, all the vertices of $Z$ that have an interview in this iteration will be removed, and we are done. When $Z$ is the long side and $e(Z)\leq \delta$, every vertex in $Z$ will be interviewed.
    
    The only remaining case to consider is where $Z$ is on the long side,  $e(Z) > \delta$, and the tiers considered alongside $v$ are $Y_1, Y_2, \ldots Y_r$, where  $e(Y_i) \leq \delta$ for all $i$. From \cref{clm:covered-indices} and the fact that $e(Z)$ decreases monotonically, as long as $v$ is not removed from the graph, it will be selected for interview, which means it will be the endpoint of at least one edge in $E_i$. That implies $v$ is in $A_i$ or $P_i$ by definition.
\end{proof}

\begin{lemma}\label{lem:positive-eps-short-side}
    Suppose $a_r \in A'_i$, $|A_i'| \leq |P_i'|$, and $|P_i'| \geq \delta$. Then, with probability $1 - O(1/n^2)$, $a_r$ gets matched in iteration $i$ of \Cref{alg:general-tiers-stable} to some position $p_j \in P'_j$ such that $\epsilon_{rj}^A > 0$. 
\end{lemma}
\begin{proof}

    Let $X$ and $Y$ be two tiers that vertices of $A_i$ and $P_i$ belong to, respectively. Also, by the definition of $A_i'$, we have $|A_i'| = e(X)$ at the time of $i$th iteration. If $e(X) \leq \delta$, since $|P_i'|\geq \delta$, we choose a subset $S$ of $Y$ with size $\delta + |Y| - e(Y)$ and interview all edges. This set contains $\delta$ vertices that are not matched at $i$th iteration. By \Cref{lem:final-bound-single-tier}, if we choose $t = 10\log^2 n /\log^2 e(Y)$, then $c_t>2\log n/\log e(Y)$ and with probability 
    \begin{align*}
        1 - \exp(-\delta/36) - O\left(e(Y)^{-2\log n/\log e(Y)}\right) \geq 1 - O\left(\frac{1}{n^2}\right),
    \end{align*}
    $a_r$ will match to a position $p_j \in P_i'$ such that $\epsilon_{rj}^A > 0$.

    Similarly, if $e(X) > \delta$, we choose a subset $S$ of $Y$ that contains $e(X)$ unmatched vertices. For each vertex of $A_i'$ we interview $\delta$ random positions in $P_j'$. Hence, if we choose $t=10\log^2 n/\log^2 e(X)$ in \Cref{lem:final-bound-single-tier}, then $c_t>2\log n/\log e(X)$, and with probability
    \begin{align*}
        1 - \exp(-\delta/36) - O\left(e(X)^{-2\log n/\log e(X)}\right) \geq 1 - O\left(\frac{1}{n^2}\right),
    \end{align*}
    $a_r$ will match to a position $p_j \in P_i'$ such that $\epsilon_{rj}^A > 0$.
\end{proof}

\begin{lemma}\label{lem:all-vertices-get-matched}
With probability $1-O(1/n^2)$, if a vertex is removed in iteration $i$ of \Cref{alg:general-tiers-interview}, it will get matched via an edge in $\cup_{j \leq i} E_j$.  
\end{lemma}
\begin{proof}

    Without loss of generality, assume that the removed vertex in iteration $i$ is applicant $a_r$ in tier $X$. Let $Y$ be the tier of the other side that the algorithm considers in iteration $i$. First, assume that $X$ is the short side, $e(X) \leq \delta$, and $e(Y) \leq \delta$ at iteration $i$. According to Case 1 of \Cref{alg:general-tiers-interview}, all vertices of $X$ interview all vertices of $Y$, hence, the statement hold with probability 1. Now suppose that $X$ is the short side and $e(Y) > \delta$ at iteration $i$. Hence, by  definition, we have $|A_i'| \leq |P_i'|$ and $|P_i'| > \delta$. Therefore, by \Cref{lem:positive-eps-short-side}, $a_r$ is matched to a position $p_j$ such that $\epsilon_{ij}^A > 0$ with probability $1-O(1/n^2)$, which implies that $(a_r, p_j) \in E_i$. Furthermore, if $e(X) > \delta$, $e(Y) > \delta$, and $X$ is the long side,  $|A_i'| = |P_i'|$. Since for this case we run a position-proposing Gale-Shapley in \Cref{alg:general-tiers-stable}, by \Cref{lem:positive-eps-short-side}, all positions in $P_i'$ get matched to applicants in $A_i'$ with edges $E_i$. Thus, all vertices of $A_i'$ are matched to positions in $P_i'$ using edges $E_i$ since we have $|A_i'|=|P_i'|$.

    The only case that remains to be investigated is when $X$ is the long side in several iterations until $a_r$ is removed at iteration $i$, i.e. $a_r \in A_{j}$, $|A_j'| \geq |P_j'|$ for $j \in [l, i]$ and $|P_j'| \leq \delta$. Also, let $Y_l, Y_{l+1}, \ldots, Y_i$ be the tiers of the other side in each iteration. Let $C_r = \sum_{l \leq j \leq i} e(Y_j)$. First, we prove that $C_r \geq \theta$.

    
    Let $X^{j-1}$ be tier $X$ before the $j$th iteration of \Cref{alg:general-tiers-interview} for $j \in [l,i+1]$. Also, let $\psi$ be the relative position of $a_r$ in tier $X$ before any iteration. During iteration $j\in[l,i]$, since $|e(Y_j)|\leq \delta$, \Cref{clm:covered-indices} shows that (1) every vertex in $X^j$ with index at most $S(X)+\delta-e(X^{j-1})$ is in $A_j$; (2) $e(X^j)=e(X^{j-1})-e(Y_j)$. For $X^{l-1}$, the first case is that $e(X^{l-1})=|X^{l-1}|$ and $S(X)-e(X^{l-1})+1\leq \psi$ when there is no iteration on $X$ before or when $|Y_{l-1}| > \delta$, and the second case is that $\psi>S(X)+\delta-e(X^{l-2})\geq S(X)+e(Y_{l-1})-e(X^{l-2})\geq S(X)-e(X^{l-1})$ when $|Y_{l-1}|\leq \delta$. In both cases, we have $\psi\geq S(X)-e(X^{l-1})+1$. Since $a_r$ is removed in iteration $i$, we have $S(X)-(e(X^i)+\theta)+1\geq \psi$, thus $C_r=\sum_{l\leq j\leq i}e(Y_j)=e(X^{l-1})-e(X^i)\geq \theta.$

    Now consider iteration $j$ and suppose that $a_r$ is not matched yet. Then, with a probability of at least $|P_j'|/\delta$, applicant $a_r$ will be matched in iteration $j$ because all applicants of $A_j'$  are in the same tier, and subjective interests are i.i.d. Moreover, $|P_j'| = e(Y_j)$ by  definition. Therefore, the probability that $a_r$ remains unmatched after iteration $i$ is at most
    \begin{align*}
        \prod_{j=l}^i\left(1-\frac{|P_j'|}{\delta}\right) = \prod_{j=l}^i\left(1-\frac{e(Y_j)}{\delta}\right) \leq \exp\left(-\frac{C_r}{\delta}\right) \leq \exp\left(-\frac{\theta}{\delta}\right),
    \end{align*}
    which completes the proof because of our choice of $\delta$ and $\theta$.
\end{proof}

Using a similar argument as \Cref{lem:all-vertices-get-matched}, the next lemma gives a bound on the degree of every node in $G$.

\begin{lemma}
    The number of interviews assigned to every position or applicant is at most $O(\log^3 n)$
\end{lemma}
\begin{proof}
    We will show the above for every applicant $a_r$. The proof for positions is the same.  Let $X$ be the tier that $a_r$ belongs to. If $X$ is the tier of the short side at iteration $i$, then $a_r$ has at most $O(\theta)$ incident edges in $E_i$ by definition of \Cref{alg:general-tiers-interview} and \Cref{inv:few-not-expired}. Also, if both $|A_i'| > \delta$ and $|P_i'| > \delta$, by Case 3 of the \Cref{alg:general-tiers-interview} and \Cref{inv:few-not-expired}, each vertex in $A_i$ or $P_i$ has at most $O(\theta)$ interviews in iteration $i$. Note that each vertex is in one of the above scenarios at most in one iteration since after that it gets removed. Therefore, similar to the proof of the previous lemma, it only remains to show that the number of interviews is bounded when $X$ is the long side in several iterations until $a_r$ is removed at iteration $i$, i.e. $a_r \in A_j$, $|A_j'| \geq |P_j'|$ for $j \in [l, i]$, and $|P_j'| \leq \delta$. Let $Y_l, Y_{l+1}, \ldots, Y_i$ be the tiers of the other side in each iteration.
    
    Define $C_r$, $\psi$, and $X^j$ for $j \in [l-1,i]$ as before. With a similar argument, we have $\psi \leq  S(X)+\delta-e(X^{l-1})$, and 
    \begin{align*}
        \psi \geq S(X)-|X^{i-1}|+1 &\geq S(X)-e(X^{i-1})-(|X^{i-1}|-e(X^{i-1}))+1 \\
        & \geq S(X)-e(X^{i-1})-\theta \\
        & \geq S(X)-e(X^{i})-(e(X^i)-e(X^{i-1}))-\theta \\
        & \geq S(X)-e(X^i)-(\theta+\delta)
    \end{align*} 
    Thus, $C_r \leq e(X^{l-1})-e(X^i)\leq \theta+2\delta$. Note that the number of interviews for $a_r$ is equal to $\sum_{i\leq j\leq i} |Y_j|$. For all tiers except $Y_l$, we have $e(Y_j) = |Y_j|$ at the time of the iteration $j$. Also, by \Cref{inv:few-not-expired}, we have $|Y_l| \leq e(Y_l) + \theta$. Therefore,
    \begin{align*}
        \sum_{i\leq j\leq i} |Y_j| \leq e(Y_l) + \theta + \sum_{l < j \leq i}e(Y_j) = C_r + \theta \leq 2(\theta + \delta),
    \end{align*}
    which implies the total number of interviews for vertex $a_r$ is at most $O(\theta) = O(\log^3 n)$ before it gets removed which concludes the proof.
\end{proof}

\begin{lemma}\label{lem:stability-general-tiers}
    The matching produced by the algorithm is an interim stable matching with probability $1-O(1/n)$.
\end{lemma}
\begin{proof}
    By \Cref{obs:all-removed}, all applicants and positions are removed from the graph after \Cref{alg:general-tiers-stable} finishes. Using union bound for all vertices in \Cref{lem:all-vertices-get-matched}, with probability $1 - O(1/n)$, all vertices are matched using an edge in $E$. Consider tiers $X$ and $X'$ from applicants and tiers $Y$ and $Y'$ from positions and suppose that  $X$ is a higher tier than $X'$ and $Y$ is higher than $Y'$.  Let $H_1 = X \times Y'$ and $H_2 = X' \times Y$. We claim that at most one of the following holds: $H_1 \cap E \neq \emptyset$ or $H_2 \cap E \neq \emptyset$. Without loss of generality, suppose that $H_1 \cap E \neq \emptyset$. This implies that at some iteration $i$ of the \Cref{alg:general-tiers-interview}, the algorithm considers tiers $X$ and $Y'$ in \Cref{ln:top-tiers-selection}. Hence, all vertices of $Y$ should be already removed, which implies $H_2 \cap E = \emptyset$.

    Now, consider applicant $a_i$ and position $p_j$ that are not matched to each other. We prove that $(a_i, p_j)$ is not a blocking pair. Let $\mu(a_i)$ and $\mu(p_j)$ be the matches of $a_i$ and $p_j$, respectively. Also, let $X$ and $Y$ be the tiers that include $a_i$ and $p_j$, respectively. We use $X \succ X'$ to show tier $X$ is preferable to tier $X'$. If $\mu(a_i) \in Y'$ such that $Y' \succ Y$, then this pair is not a blocking pair since $\mu(a_i) \succ_{a_i} p_j$. A similar argument works for $\mu(p_j)$. By the above argument, $\mu(a_i)$ and $\mu(p_j)$ does not belong to tiers $Y'$ and $X'$ such that $X \succ X'$ and $Y \succ Y'$. Hence, either or both $\mu(a_i) \in Y$ and $\mu(p_i) \in X$. Thus, $X$ and $Y$ are once considered at the same time in one iteration in \Cref{ln:top-tiers-selection} of \Cref{alg:general-tiers-interview}.
    
    Without loss of generality, suppose that at iteration $r$,  the algorithm considers $X$ and $Y$ in \Cref{ln:top-tiers-selection} and $X$ is the short side. By definition, we have $a_i \in A_r'$. If $p_j \notin P_r'$, then $|P_r'| \geq \delta$, which implies that $\epsilon^A_{i\mu(a_i)} > 0$. Thus, $(a_i, p_j)$ is not a blocking pair. If $p_j \in P_r'$, since the algorithm finds a stable matching between $A_r'$ and $P_r'$, then $(a_i, p_j)$ is not a blocking pair, which finishes the proof.
\end{proof}

    
\begin{theorem}\label{thm:non-adaptive-final}
    There exists a non-adaptive algorithm that finds an interim stable matching for tiered ranking agents with high probability such that the number of interviews done by each applicant or position is at most $O(\log^3 n)$.
\end{theorem}
\begin{proof}
    By \Cref{lem:stability-general-tiers}, the matching returned by the algorithm is interim stable. Furthermore, by \Cref{lem:all-vertices-get-matched}, each applicant or position has at most $O(\log^3 n)$ interviews.
\end{proof}

\subsection{Extension}\label{sec:extension}
In the aforementioned tiered market scenario, we presume the absence of ex-ante preferences within a tier. However, agents may have distinct individual preferences within tiers. For example, in the medical context preferences can be shaped by various factors, including geographical location, the availability of extracurricular activities, specific research laboratories they aspire to join, etc. 
To capture this characteristic of the market, we introduce the following model. Applicant $a_i$'s utility $v_{ij}$ for position $p_j$ is now defined to be $v_{ij}=v_j+\eta^A_{ij}+\varepsilon^A_{ij}$, where $v_j=v_k$ if and only if $p_j$ and $p_k$ are in the same tier, and $v_j>>v_k$ if $p_j$ is in a higher tier; $\eta^A_{ij}$ is a small random perturbation of the ex-ante preference which is known before the interview; For now, we assume that $\eta^A_{ij}$s and $\varepsilon^A_{ij}$s are all drawn from the uniform distribution over [-1,1] independently. The same applies to positions.

We first briefly argue that our algorithm works for the extended tiered markets. Consider the first subproblem, Single Tier Structure.
\begin{itemize}
    \item Using the same Chernoff bound arguments, we can prove that for each applicant, there is a constant $c$ such that $\delta/c$ of the interviewed position has $\eta_{ij}^A+\varepsilon^A_{ij}>1$.
    \item In the extended model, the preference lists are still random permutations over the randomness of both $\eta$ and $\varepsilon$, thus each applicant will be matched to a position with a high rank.
\end{itemize}
Combine the above facts, we can give a similar argument as \cref{lem:final-bound-single-tier} that each applicant gets matched to an interviewed position with $\eta_{ij}^A+\varepsilon_{ij}^A>1$. In this way, when we use it as a subroutine, it still satisfies the property that the applicant would never want to switch to positions that are not interviewed in the same tier since $\eta^A_{ik}<1$.

Consider the second subproblem, One Large Tier v.s. Multiple Singleton Tiers. Using a similar Chernoff argument as above, it is guaranteed that a higher-tier applicant is matched to the most preferred available position that is interviewed with $\eta_{ij}^A+\varepsilon_{ij}^A>1$. Thus they have no incentive to switch.

The proof for the extended general case is a combination of the two extended subproblems, and can be modified from the current one by changing the arguments in \cref{lem:positive-eps-short-side} from $\varepsilon_{ij}^A>0$ to $\eta_{ij}^A+\varepsilon_{ij}^A>1$.

\paragraph*{Remark} Although the boundedness of uniform distribution plays an important role in the analysis, more general distributions work, as long as with constant probability $\eta_{ij}^A+\varepsilon_{ij}^A$ is larger than the maximum of $n$ i.i.d.\ $\eta_{ik}^A$s.

\section{Conclusion}
The extensive body of work on two-sided matching has given rise to a comprehensive theory and efficient implementations in real-world scenarios. The majority of these studies assume that agents are fully aware of their preferences. This leaves a critical aspect—the interaction period where agents learn about their preferences via mutual interaction—largely under-explored. The current paper ventures into this area, using an algorithmic lens to extend Gale-Shapley's deferred acceptance to this setting.

We present an algorithmic approach that extends Gale-Shapley's deferred acceptance to include situations where agents form preferences during the interaction period. We propose two algorithms. The first, an adaptive algorithm, integrates interviews between applicants and positions into Gale-Shapley's deferred acceptance. Like deferred acceptance, one side sequentially proposes to the other. However, the order of proposals is arranged to increase the chances of achieving an interim stable match. Further, our analysis shows that the number of interviews conducted by each applicant or position is, with high probability, limited to $O(\log^2 n)$.

We also propose a non-adaptive algorithm for markets where the dynamics consist of an initial interview phase followed by a clearing stage. In these situations, having prearranged interview lists for each position may be necessary. Our non-adaptive algorithm creates these lists simultaneously. It aims to find a stable match in markets where applicants and positions are divided into tiers, and it limits the number of interviews for each applicant or position to no more than $O(\log^3 n)$.

We  point out that the algorithms in this work do not account for the incentives of applicants or positions, i.e.,  mechanisms induced by the algorithm don't carry the dominant strategy incentive compatible (DSIC) property as the Gale-Shapley's algorithm. 




\section*{Acknowledgment} 

The authors acknowledge the support of NSF awards 2209520, 2312156, and a gift from CISCO.

\bibliographystyle{plainnat}
\bibliography{references}

\appendix

\vspace{-0.3em}

\section{Deferred Proofs}

\positiveCounts*
\begin{proof}
    Denote $X_{ij}$ as the indicator random variable of whether $\epsilon_{ij}^A\geq 0$, set $P_i$ as the set of positions that applicant $a_i$ is interviewing, and $X_i=\sum_{p_j\in P_i}X_{ij}$ as the total number of interviewed positions with $\epsilon_{ij}^A\geq 0$. By definition, we know that $X_{ij}$'s are i.i.d.\ Bernoulli random variables with mean $1/2$. By Chernoff bound, we have
    \begin{align*}
        \Pr\left[X_i\leq \frac{\delta}{3}\right]\leq \Pr\left[X_i\leq \left(1-\frac{1}{3}\right)\frac{\delta}{2}\right]\leq \exp\left(-\frac{\delta}{36}\right),
    \end{align*}
    which finishes the proof.
\end{proof}

\end{document}